\documentclass{llncs}
\usepackage[utf8x]{inputenc}

\usepackage{graphicx} 
\usepackage{amssymb} 

\usepackage{mathtools}

\usepackage{algorithm}
\usepackage{algorithmic}

\DeclareMathOperator*{\argmax}{arg\,max}

\title{PTAS for Minimax Approval Voting}
\author{Jarosław Byrka\thanks{jby@cs.uni.wroc.pl},
	Krzysztof Sornat\thanks{krzysztof.sornat@cs.uni.wroc.pl}}

\institute{Institute of Computer Science, University of Wrocław\\
Joliot-Curie 15\\
50-383 Wrocław, Poland}

\begin{document}
\maketitle

\begin{abstract}
We consider Approval Voting systems where each voter decides on a subset of candidates he/she approves. We focus on the optimization problem of finding the committee of fixed size k, minimizing the maximal Hamming distance from a vote. In this paper we give a PTAS for this problem and hence resolve the open question raised by Carragianis et al. [AAAI'10]. The result is obtained by adapting the techniques developed by Li et al. [JACM'02] originally used for the less constrained Closest String problem. The technique relies on extracting information and structural properties of constant size subsets of votes. 
\end{abstract}

\section{Introduction}
Approval Voting systems are widely considered~\cite{brams_av_book} as an alternative to traditional elections, where each voter may select and support at most some small number of candidates.  
In Approval Voting each voter decides about every single candidate if he approves the candidate or does not approve him/her. 
A result is obtained by applying a predefined election rule to the set of collected votes. 

In this paper we study the problem of implementing an appropriate election rule and focus on the Minimax objective \cite{brams}: we minimize the biggest dissatisfaction over voters.
The resulting optimization problem is denoted $MAV$, and it is to select a committee composed of exactly $k$ candidates, and minimizing the maximal symmetric difference between the committee and the set of approved candidates by a single voter. 

Using the string terminology, votes are encoded as strings, and the goal is to find a string encoding a committee minimizing the maximal Hamming distance to an input string.
Unlike in the related Closest String problem, in $MAV$ there is also a constraint: the selected committee must be of fixed size $k$, and hence in the string terminology there must be exactly $k$ ones in the string.

\subsection{Related work and our results}
Many different objective functions have been proposed and studied in the context of selecting the committee based on the set of votes collected in an Approval Voting system~\cite{computations_av,brams_av_book}.
Clearly, optimizing the sum of Hamming distances to all votes is an easy task and can be done by simply selecting the $k$ candidates approved by the largest number of voters.
By contrast, Minimax Approval Voting was shown by LeGrand~\cite{nphard} to be NP-hard.
LeGrand et al.~\cite{kcompletion} obtained $3$-approximation by a very simple $k$-completion algorithm. Next, Carragianis et al.~\cite{markakis} gave the currently best $2$-approximation algorithm.
The algorithm was obtained by rounding a fractional solution to the natural LP relaxation of the problem, and obtained approximation ratio essentially matches the integrality gap of the LP. 

In this paper we give a PTAS for the Minimax Approval Voting problem. Our work is based on the PTAS for Closest String \cite{ptascs}, 
which is a similar problem to $MAV$ but there we do not have the restriction on the number of 1's in the result. 
Technically, our contribution is the method of handling the number of 1's in the output. We also believe that our presentation is somewhat more intuitive.

Approval Voting systems are also analyzed in respect of manipulability, see e.g., \cite{computations_av} or \cite{markakis}. In particular, \cite{markakis} proved that each strategy-proof algorithm for $MAV$ must have approximation ratio at least $2-\frac{2}{k+1}$, which implies that our PTAS cannot be strategy-proof.

\subsection{Definitions}
We will use the following notation:\\
$n$ -- number of voters,\\
$m$ -- number of candidates,\\
$s_i \in \{0,1\}^m$ -- a vote of voter $i$,\\
$s_i[j]=1$ if voter $i$ approves candidate $j$,\\
$s_i[j]=0$ if voter $i$ does not approve candidate $j$,\\
$S=\{s_1,s_2,\dotsc,s_n\}$ -- the set of collected votes,\\
$s^{(1)}=\big|\{j:s[j]=1\}\big|$ -- the number of 1's in $s$.\\
For $x,y \in [\,0,1]^m$ we define a distance $d(x,y) = \sum_{j=1}^m \big|x[j]-y[j]\big| = \lVert x-y \rVert _{1}$.\\
For $x,y \in \{0,1\}^m$, $d(x,y)$ is called the Hamming distance.

\begin{definition}\label{def_mavk}
\[ OPT=\min_{\substack{x \in \{0,1\}^m\\x^{(1)}=k}} \: \max_{i \in \{1,2,\dots,n\}} d(x,s_i)\] 
Let $s_{OPT}$ be an optimal solution, i.e., $\max_{i \in \{1,2,\dots,n\}} d(s_{OPT},s_i)=OPT$.
\end{definition}
WLOG we assume that $n>k$. If not, we copy the first string $k-n+1$ times.

\subsection{The main idea behind our algorithm}\label{name_consensus}
The general idea behind our PTAS is to find a small enough subset $X$ of votes that is a ``good representation'' of the whole set of votes $S$. 
Then the candidates are partitioned into those for which voters in $X$ agree and the rest of candidates. For the ``consensus candidates'' we fix our decision
to the decision induced by votes in $X$ (additionally correcting the number of selected candidates in the ``consensus'' set).
Next, we consider the optimization problem of finding a proper subset of the remaining candidates to join the committee.
The key insight is that there exists a small enough subset $X$ such that the induced decision for the ``consensus candidates'' will not be a big mistake.  

\subsection{Organization of the paper}
First, in Section~\ref{sec:info_subsets} we formalize the information we may extract from subset of votes, and introduce a measure of inaccuracy of such a subset. 
Next, in Section~\ref{sec:existence_of_subset} we prove the existence of a small subset of votes with stable inaccuracy.
In Section~\ref{sec:aux_prob} we show that the optimization problem of deciding the part of the committee not induced by the subset of votes can be approximated
with only a small additional loss in the objective function.
Finally, in Section~\ref{sec:alg} we give an algorithm considering all subsets of a fixed size and show that, in the iteration when the algorithm
happens to consider a subset with stable inaccuracy, it will produce a $(1+\epsilon)$-approximate solution to $MAV$.

\section{Extracting information from subsets}
\label{sec:info_subsets}
We consider subsets of votes and analyze the information they carry. We measure the inaccuracy of this information with respect to the set of all votes. We show that there exists a small subset with stable inaccuracy, i.e., the drop of inaccuracy after including one more vote is small.

Let us define an inaccuracy function $ina:2^S \mapsto \mathbb{R}_{\geqslant 0}$ that measures the inaccuracy if we will consider subset $Y \subseteq S$ instead of $S$. The smaller the $ina(Y)$ is the better the common parts of strings in $Y$ represent $s_{OPT}$.

\begin{definition}\label{def_ina}
 For all $Y \subseteq S, Y \neq \emptyset$ we define functions $t(Y) \in \{0,1\}^m$ and $ina(Y) \in \mathbb{R}_{\geqslant 0}$ as follows:
 \[\big(t(Y)\big)[j] =
  \begin{cases}
   0 & \text{if } \forall_{y \in Y} \quad y[j]=0\\
   1 & \text{if } \forall_{y \in Y} \quad y[j]=1\\
   s_{OPT}[j] & \text{otherwise,}
  \end{cases}\]
 \[ ina(Y) = d(t(Y),s_{OPT}).\]
\end{definition}

Intuitively $t(Y)$ is the optimal solution $s_{OPT}$ changed at positions where all strings from $Y$ agree. Also we define the pattern of a subset of votes.
\begin{definition}\label{def_p}
 For all $Y \subseteq S, Y \neq \emptyset$ we define pattern $p(Y) \in \{0,1,*\}^m$ as:
 \[ \big(p(Y)\big)[j]=
  \begin{cases}
   0 & \text{if } \forall_{y \in Y} \quad y[j]=0\\
   1 & \text{if } \forall_{y \in Y} \quad y[j]=1\\
   * & \text{otherwise.}
  \end{cases}
 \]
\end{definition}
It represents positions that all strings in $Y$ agree. ``$*$'' encodes a mismatch. 

Note that (from Definitions \ref{def_ina} and \ref{def_p}) $t(Y)$ is an optimal solution $s_{OPT}$ overwritten by a pattern $p_r$ on no-star positions:
 \[\big(t(Y)\big)[j] =
  \begin{cases}
   s_{OPT}[j]   & \text{if } \big(p(Y)\big)[j] = * \\
   \big(p(Y)\big)[j] & \text{otherwise.}
  \end{cases}
 \]

The inaccuracy function has the following properties:
\begin{lemma}\label{decreasing_ina}
 $\forall_{s_{i_1} \in S}$, for all sequences $\{s_{i_1}\}=Y_1 \subseteq Y_2 \subseteq \dots \subseteq Y_n = S$ we have
 \[OPT \geqslant ina(Y_1) \geqslant ina(Y_2) \geqslant \dots \geqslant ina(Y_n) = 0\]
\end{lemma}

\begin{proof}
 
 It is easy to see that
 \[ina(Y_1) \stackrel{\text{def.}}{=} d\big(t(Y_1),s_{OPT}\big)=d\big(t(\{s_{i_1}\}),s_{OPT}\big)=d\big(s_{i_1},s_{OPT}\big) \leqslant OPT,\]
 \[ina(Y_n) = ina(S) = d(s_{OPT},s_{OPT}) = 0.\]
 Still we need to prove $ina(Y_i) \geqslant ina(Y_{i+1})$. Pattern $p(Y_{i+1})$ is built on strings from $Y_i \subseteq Y_{i+1}$ and strings from $Y_{i+1} \setminus Y_i$. So $p(Y_{i+1})$ has at least as many $*$ as $p(Y_i)$ has. Therefore $t(Y_{i+1})$ has at least as many positions as $t(Y_i)$ has that agree with optimal solution $s_{OPT}$, so $d\big(t(Y_i),s_{OPT}\big) \geqslant d\big(t(Y_{i+1}),s_{OPT}\big)$. Using definition of the inaccuracy function (Definition \ref{def_ina}) we prove the lemma. $\hfill\blacksquare$
\end{proof}

Intuitively $ina(Y)-ina(Y \cup \{y\})$ is the decrease of the inaccuracy from adding element $y$ to set $Y$. We will show that, when adding one more element $y$ to sets $Y,Z$ such that $Y \subseteq Z$, the inaccuracy decrease more in a case of adding $y$ to the smaller set $Y$ than adding $y$ to the bigger set $Z$.
\begin{lemma}\label{supermodular_ina}
 If we artificially extend the $ina(\cdot)$ function for the empty set:\\ $ina(\emptyset)=2\cdot OPT$, then function $ina(\cdot)$ is supermodular\footnote{according to \cite{schrijver}, $f:2^S\mapsto \mathbb{R}$ is supermodular iff\\ $\forall_{Y,Z \subseteq S} \quad f(Y)+f(Z)\leqslant f(Y \cup Z)+f(Y\cap Z)$ which is equivalent with $\forall_{Y \subseteq Z \subseteq S} \quad\forall_{s\in S} \quad f(Z)-f(Z\cup\{s\})\leqslant f(Y)-f(Y \cup\{s\})$.}, i.e.,
 \begin{equation}\label{supermodular_ina_ieq}
  \forall_{Y \subseteq Z \subseteq S} \quad\forall_{s \in S} \quad ina(Z)-ina(Z \cup \{s\}) \leqslant ina(Y)-ina(Y \cup \{s\})
 \end{equation}
\end{lemma}

\begin{proof}
 Let fix $Y,Z$ and $s$ such that $Y \subseteq Z \subseteq S$ and $s \in S$.
 
 \paragraph{\textbf{Case 1:}} $Z = \emptyset$:
 
 Then also $Y = \emptyset$, and inequality (\ref{supermodular_ina_ieq}) holds obviously.
 
 \paragraph{\textbf{Case 2:}} $Z \neq \emptyset, Y=\emptyset$:
 
 We have:
 \[ina(Z)-ina(Z \cup \{s\}) \leqslant OPT = \]
 \begin{equation}\label{supermodular_ina_c2}
  = 2\cdot OPT-OPT \leqslant ina(\emptyset)-ina(\{s\}) = ina(Y)-ina(Y \cup \{s\}),  
 \end{equation}
 because we use respectively: Lemma \ref{decreasing_ina} and the fact that $Z$ has at least one element; definition of $ina(\cdot)$ for empty set and upperbound for $ina(\cdot)$ function; assumption that $Y = \emptyset$.
 
 \paragraph{\textbf{Case 3:}} $Z \neq \emptyset, Y \neq \emptyset$:
 
 From definition of $ina(\cdot)$ we have:
 \[ ina(Z)-ina(Z \cup \{s\}) = d\big(t(Z),s_{OPT}\big)-d\big(t(Z \cup \{s\}),s_{OPT}\big) = \]
 counting a difference by considering two cases for value of $s_{OPT}$ we obtain
 \[ = \Big| \big\{ j: s_{OPT}[j]=1 \wedge t(Z \cup \{s\})[j]=1 \wedge t(Z)[j]=0 \big\}\Big| + \]
 \[ +\;\, \Big| \big\{ j: s_{OPT}[j]=0 \wedge t(Z \cup \{s\})[j]=0 \wedge t(Z)[j]=1 \big\}\Big| = \]
 using definition of function $t(\cdot)$:
 \[ = \Big| \big\{ j: s_{OPT}[j]=1 \wedge s[j]=1 \wedge \;\forall_{z \in Z} \; z[j]=0 \big\}\Big| + \]
 \[ +\;\, \Big| \big\{ j: s_{OPT}[j]=0 \wedge s[j]=0 \wedge \;\forall_{z \in Z} \; z[j]=1 \big\}\Big| \leqslant \]
 taking an universal quantifier over a smaller subset we obtain: 
 \[ \leqslant \Big| \big\{ j: s_{OPT}[j]=1 \wedge s[j]=1 \wedge \;\forall_{y \in Y} \; y[j]=0 \big\}\Big| + \]
 \[ + \;\,\Big| \big\{ j: s_{OPT}[j]=0 \wedge s[j]=0 \wedge \;\forall_{y \in Y} \; y[j]=1 \big\}\Big| = \]
 reversing all previous transformations finally we obtain:
 \[ = ina(Y)-ina(Y \cup \{s\}). \]$\hfill\blacksquare$
\end{proof}

\section{Existence of a stable subset}
\label{sec:existence_of_subset}
\begin{lemma}\label{lem_exists_x}
For any fixed $R\in\mathbb{N}_{\geqslant 1}$ there exists a subset $X \subseteq S, |X| = R$ such that
\begin{equation}\label{lem_x_exists_ina}
 \forall_{s \in S \setminus X} \quad ina(X) - ina(X \cup \{s\}) \leqslant \frac{OPT}{R}.
\end{equation}
We say such $X$ is $\frac{OPT}{R}$-stable.
\end{lemma}

It means that there exists such a subset of votes $X$ that adding one more vote into $X$ the inaccuracy decreases by at most $\frac{OPT}{R}$.

\begin{proof}
 First, we construct $S_{\underline{r}}$ satisfying (\ref{lem_x_exists_ina}) with at most $R$ elements.
 
 Let us construct a sequence of subsets $S_1 \subset S_2 \subset \dotsc \subset S_n = S, |S_i|=i$. We take $S_1=\{s_{i_1}\}$, where $s_{i_1}$ is any element of $S$ and for $r \in \{2,3,\dots,n\}$ we take $S_{r}=S_{r-1}\cup\{s_{i_r}\}$ where $s_{i_r}$ is such a vote that after adding it the~inaccuracy function decreases the most, i.e.,
 \begin{equation}\label{thm_sir}
  s_{i_r} = \argmax_{s \in S\setminus S_{r-1}} \big(ina(S_{r-1})-ina(S_{r-1} \cup \{s\})\big).  
 \end{equation}
 \begin{figure}\label{fig:ina}
  \centering
  \includegraphics[scale=1.15]{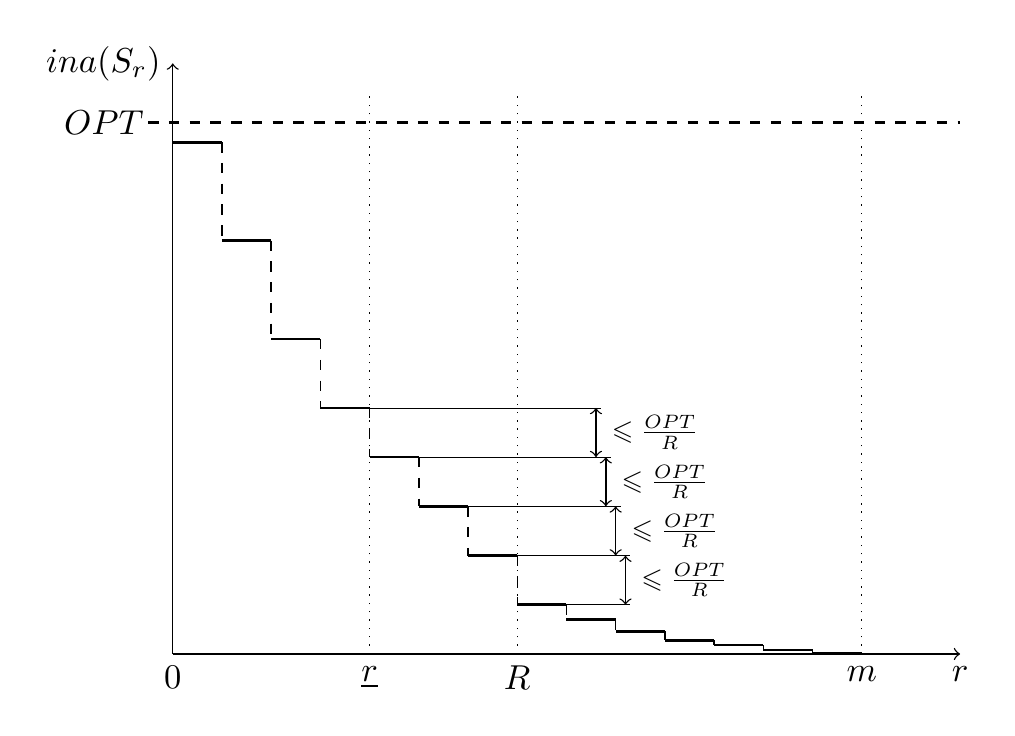} 
  \vspace{-10pt}
  \caption{The $ina(\cdot)$ function for the sequence of subsets $S_1 \subset S_2 \subset \dotsc \subset S_n = S$.}
 \end{figure}
 We have
 \[ \min_{r \in \{1,2,\dots,R\}} \quad ina(S_r)-ina(S_{r+1}) \leqslant \frac{1}{R} \left(\sum_{r=1}^R ina(S_r)-ina(S_{r+1})\right) = \]
 \begin{equation}\label{thm_min_r}
 = \frac{1}{R} \big( ina(S_1)-ina(S_{R+1}) \big) \leqslant \frac{OPT}{R},  
 \end{equation}
 because (from Lemma \ref{decreasing_ina}) we know that $ina(S_1) \leqslant OPT$ and $ina(S_{R+1}) \geqslant 0$. Let $\underline{r}$ be a minimizer for the left-hand side of (\ref{thm_min_r}), then (by the choice of $s_{i_{\underline{r}}}$ in (\ref{thm_sir})) we have:
 \begin{equation}\label{thm_max_diff_ina}
  \max_{s \in S\setminus S_{\underline{r}}} \big(ina(S_{\underline{r}})-ina(S_{\underline{r}} \cup \{s\})\big)\leqslant \frac{OPT}{R},
 \end{equation}
 thus $S_{\underline{r}}$ satisfies (\ref{lem_x_exists_ina}), see Figure 1. If $S_{\underline{r}}$ has less elements than $R$ we can extend $S_{\underline{r}}$ to an $R$-elements subset $X$ by adding any elements of $S$. It follows from the supermodularity of $ina(\cdot)$. From Lemma \ref{supermodular_ina} we have:
 \[\forall_{s \in S \setminus S_{\underline{r}}} \quad ina(X)-ina(X \cup \{s\}) \leqslant ina(S_{\underline{r}})-ina(S_{\underline{r}} \cup \{s\}), \]
 and hence also:
 \begin{equation}\label{thm_max_sx_x}
  \max_{s \in S \setminus S_{\underline{r}}} \big(ina(X)-ina(X \cup \{s\})\big) \leqslant \max_{s \in S \setminus S_{\underline{r}}} \big(ina(S_{\underline{r}})-ina(S_{\underline{r}} \cup \{s\})\big).
 \end{equation}
 Finally, taking (\ref{thm_max_diff_ina}) and (\ref{thm_max_sx_x}) we obtain:
 \[\max_{s \in S \setminus X} \big(ina(X)-ina(X \cup \{s\})\big) \leqslant \frac{OPT}{R}.\]$\hfill\blacksquare$
\end{proof}
Of course we cannot construct such a subset efficiently if we do not know $s_{OPT}$. How to find a proper subset $X$? For constructing our PTAS we will fix $R \in \mathbb{N}_{\geqslant 1}$ and consider all subsets $Y \subseteq S$ with cardinality $R$. There is less than $n^R \in Poly(n)$ such subsets. For clarity, we will use $Y \subseteq S$ in arguments valid for all subsets considered by the algorithm, and $X \subseteq S$ for a $\frac{OPT}{R}$-stable subset of votes. 

For a fixed  $Y \subseteq S, Y \neq \emptyset$, WLOG we reorder candidates in such a way that $p(Y)$ is a lexicographically smallest permutation:
\[p(Y)=**\dotsc *00\dotsc 011\dotsc 1.\]
The first part (from the left) is called ``star positions'' or ``star part''. The remaining part is called ``no-star part''. We define $p^{(*)}(Y)$ as the number of $*$ in $p(Y)$ and we denote it $\beta$:
\[ \beta = p^{(*)}(Y) = \Big|\left\{ j:\big(p(Y)\big)[j]=* \right\}\Big|.\]

In our PTAS we essentially fix the ``no-star part'' of the answer to the pattern $p(Y)$ and optimize over the choices for the ``star part'' of the outcome. If the~number of stars or number of 1's on star positions of $s_{OPT}$ is small enough, then there is only $Poly(m,n)$ possible solutions and we can consider all of them. Let us analyze the size of the ``star part''.
\begin{lemma}\label{size_p_star_y}
 For all $Y \subseteq S$ we have
 \[\beta = p^{(*)}(Y) \leqslant |Y| \cdot OPT\]
\end{lemma}
\begin{proof}
   Consider an arbitrary $Y=\{y_1,y_2,\dots,y_{|Y|}\}$. We can construct $Y$ in the~following 3 phases:
   \begin{enumerate}
   \item[1.] $Y := \{s_{OPT}\}$
   \item[2.] for $i=1$ to $|Y|$ do
   \item[] \quad $Y := Y \cup \{y_i\}$
   \item[3.] $Y := Y \setminus \{s_{OPT}\}$
   \end{enumerate}
   After that we obtain set Y. Let us calculate how many stars $p(Y)$ has. In Phase~1 there are no stars. In each step in Phase~2 we add at most $OPT$ stars, because $\forall_{i\in \{1,2,\dots,|Y|\}} \quad d(y_i,s_{OPT}) \leqslant OPT$. In Phase 3 we can at most decrease the~number of stars. So $\beta \leqslant |Y| \cdot OPT$.$\hfill\blacksquare$
\end{proof}

Note that for $X$ from Lemma \ref{lem_exists_x} we have
\begin{equation}\label{size_p_star} 
 p^{(*)}(X) \leqslant |X| \cdot OPT = R \cdot OPT.
\end{equation}

Let us now introduce some more notation. Assuming $Y \subseteq S$ and hence also $\beta = p^{(*)}(Y)$ are fixed,
we will use the following notation to denote the ``star part'' and the ``no-star part'' of a string $x \in \{0,1\}^m$: 
\[x'=x[1]\cdot x[2]\cdot \dotsc\cdot x[\beta],\]
\[x''=x[\beta+1]\cdot x[\beta+2]\cdot\dotsc\cdot x[m],\]
where``$\cdot$'' is a concatenation of strings (letters).
So we divide $x$ into two parts: $x=x'\cdot x''$. 

Let us now define a $k$-completion of $x \in \{0,1\}^m$ (definition from \cite{kcompletion}) to be a $y \in \{0,1\}^m$ such that $y^{(1)}=k$ and $d(y,x)$ is the minimum possible Hamming distance between $x$ and any vector with $k$ of 1's. To obtain a $k$-completion we only add or only delete a proper number of 1's. To be more specific in this paper we assume the $k$-completion is always obtained by changing bits at positions with the smallest possible index\footnote{Any other deterministic rule would work for us just as well.}.

In the following lemma we will show that for the pattern from a stable subset $X$ we can change the number of 1's in the ``no-star part'' to the properly guessed number of 1's loosing only twice the stability constant.
\begin{lemma}\label{lemma_kbiscompletion}
 If $X \subseteq S$ is $(\epsilon_1 \cdot OPT)$-stable, $z''$ is a $k''$-completion of $\big(p(X)\big)''$, where $k''=(s_{OPT}'')^{(1)}$, then
 \begin{equation}\label{lemma_kbis_completion}
  \forall_{i \in \{1,2,\cdots,n\}} \quad d(s_{OPT}'\cdot z'',s_i) \leqslant (1+2\epsilon_1) \cdot OPT
 \end{equation}
\end{lemma}
\begin{proof}
 WLOG there is insufficient number of 1's in no-star part of pattern $p(X)$, i.e., $k'' \geqslant \big((p(X))''\big)^{(1)}$. The other case is symmetric.
 
 Let us fix $s_i \in S$ and consider all combinations of values in strings $\big(p(X)\big)''$, $z''$, $s_i''$, $s_{OPT}''$ at the same position $j$. $\alpha_a \in \mathbb{N}$, for $a\in \{1,2,\cdots,12\}$, counts the number of positions $j$ with combination $a$, see Table 1.
 \begin{table}[h]\label{table_cases}
 \centering
 \begin{tabular}{|c|c|c|c|c|c|c|c|c|c|c|c|c|}
 \hline
 & \multicolumn{12}{c|}{combinations}            \\ \hline
 \quad index of a combination \quad  & 1 & 2 & 3 & 4 & 5 & 6 & 7 & 8 & 9 & 10& 11& 12\\ \hline
 $(p(X))''[j]$  & 0 & 0 & 0 & 0 & 0 & 0 & 0 & 0 & 1 & 1 & 1 & 1 \\ 
 $z''[j]$       & 0 & 0 & 0 & 0 & 1 & 1 & 1 & 1 & 1 & 1 & 1 & 1 \\ 
 $s_i''[j]$     & 0 & 1 & 0 & 1 & 0 & 1 & 0 & 1 & 0 & 1 & 0 & 1 \\ 
 $s_{OPT}''[j]$ & 0 & 0 & 1 & 1 & 0 & 0 & 1 & 1 & 0 & 0 & 1 & 1 \\ \hline
 number of occurrences & $\alpha_1$ & $\alpha_2$ & $\alpha_3$ & $\alpha_4$ & $\alpha_5$ & $\alpha_6$ & $\alpha_7$ & $\alpha_8$ & $\alpha_9$ & $\alpha_{10}$ & $\alpha_{11}$ & $\alpha_{12}$ \\ \hline
 $d(z''[j],s_i''[j])$     & 0 & 1 & 0 & 1 & 1 & 0 & 1 & 0 & 1 & 0 & 1 & 0 \\ \hline
 $d(s''_{OPT}[j],s_i''[j])$ & 0 & 1 & 1 & 0 & 0 & 1 & 1 & 0 & 0 & 1 & 1 & 0 \\ \hline
 \end{tabular}
 \vspace{5pt}
 \caption{Combinations of values in strings $(p(X))'', z'', s_i'', s_{OPT}''$. There is only 12 combinations (no $2^4=16$), because by the assumption $k'' \geqslant ((p(X))'')^{(1)}$ we never change from 1 in $((p(X))'')^{(1)}$ to 0 in $z''$.}
 \vspace{-15pt}
 \end{table}
 
 We have:
 \[ d(z'',s_i'') = \big|\{j: z''[j] \neq s_i''[j] \}\big| = \]
 we consider two cases for value of $s_{OPT}$ at position $j$:
 \[ = \big|\{j: z''[j] \neq s_i''[j] \wedge (z''[j] = s_{OPT} \vee z''[j] \neq s_{OPT})\}\big| = \]
 we divide it into two components:
 \begin{alignat*}{4}
   = &\big|\{j: s_{OPT} = && \; z''[j] \neq s_i''[j]            &&\}\big| + \\
   + &\big|\{j:           && \; z''[j] \neq s_i''[j] = s_{OPT}) &&\}\big| = 
 \end{alignat*}
 we use case counts from Table 1 to count positions in both components:
 \[ = ( \underbrace{\alpha_2 + \alpha_7 + \alpha_{11}}_\text{first component} + \underbrace{\alpha_3 + \alpha_6 + \alpha_{10} - \alpha_3 - \alpha_6 - \alpha_{10}}_{=0}) + \underbrace{(\alpha_4 + \alpha_5 + \alpha_9)}_\text{second component} = \]
 and we use the definition of the Hamming distance:
 \begin{equation}\label{d_z_bis_s_i_bis}
 = \big(d(s_{OPT}'',s_i'') - \alpha_3 - \alpha_6 - \alpha_{10}\big) + (\alpha_4 + \alpha_5 + \alpha_9).
 \end{equation}
 Since $(z'')^{(1)} = k'' = \big(s_{OPT}''\big)^{(1)}$,
 \[ \sum_{k=5}^{12} \alpha_k = \alpha_3 + \alpha_4 + \alpha_7 + \alpha_8 + \alpha_{11} + \alpha_{12} \]
 \begin{equation}\label{alpha5}
  \alpha_5 = \alpha_3 + \alpha_4 - \alpha_6 - \alpha_9 - \alpha_{10}. 
 \end{equation}
 Also
 \begin{equation}\label{epsilon1opt_stable}
  \alpha_4 + \alpha_8 + \alpha_9 \leqslant \epsilon_1 \cdot OPT,
 \end{equation}
 because $X$ is $\epsilon_1 \cdot OPT$-stable.
 Now we are ready to prove equation (\ref{lemma_kbis_completion}).
 
 \[ d(s_{OPT}'\cdot z'',s_i) \stackrel{\rm def.}{=} d(s_{OPT}',s_i') + d(z'',s_i'') \stackrel{(\ref{d_z_bis_s_i_bis})}{=} \]
 \[ \stackrel{(\ref{d_z_bis_s_i_bis})}{=} d(s_{OPT}',s_i') + d(s_{OPT}'',s_i'') - \alpha_3 - \alpha_6 - \alpha_{10} + \alpha_4 + \alpha_5 + \alpha_9 \stackrel{(\ref{alpha5})}{=} \]
 \[ \stackrel{(\ref{alpha5})}{=} \underbrace{d(s_{OPT},s_i)}_{\leqslant OPT} + 2(\hspace{-10pt} \underbrace{\alpha_4}_{\stackrel{(\ref{epsilon1opt_stable})}{\leqslant} \epsilon_1 \cdot OPT} \hspace{-15pt} - \alpha_6 - \alpha_{10}) \stackrel{(\ref{epsilon1opt_stable})}{\leqslant} (1+2\epsilon_1) \cdot OPT. \]
 $\hfill\blacksquare$
\end{proof}

\section{An auxiliary optimization problem}
\label{sec:aux_prob}
In this section we will consider the optimization problem obtained after guessing the number of 1's in the two parts and fixing the ``no-star part'' of the outcome.
It has variables for all the positions of the ``star part'' and constraints for all the original votes $s_i \in S$.

Let us define the optimization problem in terms of the integer program $IP_{(\ref{ip_q})-(\ref{ip_p01})} (Y,k')$ by (\ref{ip_q})-(\ref{ip_p01}):
\begin{equation}\label{ip_q}
 \min q
\end{equation}
\begin{equation}\label{ip_no1}
 (s')^{(1)} = k'
\end{equation}
\begin{equation}\label{ip_dist}
 \forall_{i \in \{1,2,\dots,n\}} \quad d(s',s_i') \leqslant q-d(s_{ALG}'',s_i'')
\end{equation}
\begin{equation}\label{ip_qg0}
 q \geqslant 0
\end{equation}
\begin{equation}\label{ip_p01}
 \forall_{j \in \{1,2,\dots,\beta\}} \quad s'[j] \in \{0,1\}
\end{equation}
where $Y \subseteq S, k=k'+k''$, and $s_{ALG}''$ is the $k''$-completion of $(p(Y))''$. Recall that $\beta=p^{(*)}(Y)$ and $(p(Y))''$ is the ``no-star part'' of the~pattern $p(Y)$.

In the LP relaxation (\ref{ip_p01}) is replaced with:
\begin{equation}\label{lp_p01}
 \forall_{j \in \{1,2,\dots,\beta\}} \quad s'[j] \in [0,1]
\end{equation}


Constraints (\ref{ip_q})-(\ref{ip_qg0}),(\ref{lp_p01}) are linear because
\[(s')^{(1)}=\sum_{j=1}^\beta s'[j],\]
\[d(s',s_i') = \sum_{j=1}^\beta \Big(\chi(s_i'[j]=0)\cdot s'[j] + \chi(s_i'[j]=1)\cdot(1-s'[j])\Big)\] are linear functions of $s'[j]$, where $j \in \{1,2,\dots,\beta\}$.

\begin{lemma}\label{lem_aprox_ip}
 $\forall_{R \in \mathbb{N}_{\geqslant 1}, Y \subseteq S, |Y| \leqslant R, k' \in \mathbb{N}, \epsilon_2 > 0}$ we can find $(1+2\epsilon_2)$-approximation solution for $IP_{(\ref{ip_q})-(\ref{ip_p01})} (Y,k')$ by solving the $LP$ and considering at most
 \[ (3n)^{\frac{3R \ln(2)}{(\epsilon_2)^2}} + m^{\frac{3R^2\ln(6)}{(\epsilon_2)^2}} \quad\text{cases.} \] 
\end{lemma}
\begin{proof}
Let us fix constants $\epsilon_2 \in (0,\frac{1}{2})$ (for $\epsilon_2 \geqslant \frac{1}{2}$ we could use $2$-approximation from \cite{markakis}). Consider three cases:
\paragraph{\textbf{Case 1:}} $\beta \leqslant \frac{3	R\ln(3n)}{(\epsilon_2)^2}$\\
There is $2^\beta$ possibilities for $s'$.\\
\[2^\beta \leqslant 2^{\frac{3R\ln(3n)}{(\epsilon_2)^2}} = e^{\ln(3n) \frac{3R \ln(2)}{(\epsilon_2)^2}} = (3n)^{\frac{3R \ln(2)}{(\epsilon_2)^2}} \in Poly(n),\]
because $\epsilon_2$ and $R$ are fixed constants. 
So we will check (in polynomial time) all possibilities for $s'$ and we will find optimal solution for the integer program.
\paragraph{\textbf{Case 2:}} $k'\leqslant \frac{3R^2\ln(6)}{(\epsilon_2)^2}$\\
There is $Poly(m)$ possibilities for $s'$ because we can upperbound the number of possibilities of setting 1's into $\beta$ positions by:
\[ {\beta \choose  k'} \leqslant \beta^{k'} \leqslant \beta^{\frac{3R^2\ln(6)}{(\epsilon_2)^2}} \leqslant m^{\frac{3R^2\ln(6)}{(\epsilon_2)^2}} \in Poly(m), \]
because $\epsilon_2$ and $R$ are fixed constants.
\paragraph{\textbf{Case 3:}} $\;\beta > \frac{3R\ln(3n)}{(\epsilon_2)^2} \;\wedge\; k' > \frac{3R^2\ln(6)}{(\epsilon_2)^2}$\\
We denote an optimal solution of the~$IP_{(\ref{ip_q})-(\ref{ip_p01})} (Y,k')$ by $\big((s')^{IP},q^{IP}\big)$. Let us use LP relaxation and denote an optimal solution of the LP by $\big((s')^{LP},q^{LP}\big)$. Obviously we have $q^{LP} \leqslant q^{IP}$. We can solve the LP in polynomial time but we may obtain a fractional solution. We want to round it independently. We will use a randomized rounding defined by distributions on each position $j \in \{1,2,\dots,\beta\}$:
\begin{equation}\label{pse1}
 P\big(s'[j]=1\big)= (s')^{LP}[j], \quad P\big(s'[j]=0\big)= 1-(s')^{LP}[j].
\end{equation}
We can estimate the expected value of a distance to such a random solution $s'$:
\[ \forall_{i \in \{1,2,\cdots,n\}} \quad \mathbb{E}\big[ d(s',s_i') \big] \stackrel{\text{def.}}{=} \mathbb{E}\left[ \sum_{j=1}^\beta \Big|s'[j]-s_i'[j]\Big| \right] = \]
\begin{alignat*}{6}
 = \mathbb{E}\Bigg[ &\sum_{j=1}^\beta& \Big( \quad &\chi(s_i'[j]=0)\cdot s'[j] \quad &+ \quad &\chi(s_i'[j]=1)\cdot (1-s'[j]) &&\Big) \Bigg] \stackrel{\text{lin. of }\mathbb{E}}{=}  \\
 \stackrel{\text{lin. of }\mathbb{E}}{=} &\sum_{j=1}^\beta& \Big( \quad &\chi(s_i'[j]=0)\cdot \mathbb{E}\big[s'[j]\big]\quad &+ \quad &\chi(s_i'[j]=1)\cdot \mathbb{E}\big[1-s'[j]\big] &&\Big) \stackrel{(\ref{pse1})}{=} \\
 \stackrel{(\ref{pse1})}{=} &\sum_{j=1}^\beta& \Big( \quad &\chi(s_i'[j]=0)\cdot (s')^{LP}[j]\quad &+ \quad &\chi(s_i'[j]=1)\cdot \big(1-(s')^{LP}[j]\big) &&\Big) \stackrel{\text{def.}}{=}
\end{alignat*}
\begin{equation}\label{exp_d_p_si}
  \stackrel{\text{def.}}{=} d\big((s')^{LP},s_i'\big) \stackrel{(\ref{ip_dist})}{\leqslant} q^{LP}-d(s_{ALG}'',s_i'').
\end{equation}
$d(s',s_i')$ is a sum of $\beta$ independent 0-1 variables. For $\epsilon' \in (0,1)$ using Chernoff's bound \cite{motvani} we have:
\[ P\Big(d(s',s_i') \geqslant (1+\epsilon') \cdot\mathbb{E}\big[d(s',s_i')\big] \Big) \leqslant \exp\left( -\frac{1}{3}(\epsilon')^2 \cdot \mathbb{E}\big[d(s',s_i')\big] \right). \]
If we take $\epsilon' = \frac{\epsilon_2 \cdot q^{IP}}{\mathbb{E}[d(s',s_i')]}$ then we obtain:
\[ \exp\left( -\frac{1}{3} \cdot \frac{(\epsilon_2)^2 \cdot (q^{IP})^2}{\mathbb{E}\big[d(s',s_i')\big]} \right) \geqslant P\Big(d(s',s_i') \geqslant \mathbb{E}\big[d(s',s_i')\big] + \epsilon_2 \cdot q^{IP} \Big) \stackrel{(\ref{exp_d_p_si})}{\geqslant} \]
\begin{equation}\label{expgpbb}
 \stackrel{(\ref{exp_d_p_si})}{\geqslant} P\Big(d(s',s_i') \geqslant q^{LP}-d(s_{ALG}'',s_i'') + \epsilon_2 \cdot q^{IP} \Big).
\end{equation}
We want to know an upperbound for the probability that we make an error greater than $\epsilon_2 \cdot q^{IP}$ for at least one vote:
\[ P\Big(\exists_{i \in \{1,2,\dots,n\}}: d(s',s_i') \geqslant q^{LP}-d(s_{ALG}'',s_i'') + \epsilon_2 \cdot q^{IP} \Big) \stackrel{(\ref{expgpbb})}{\leqslant} \]
\begin{equation}\label{pbb_error_all}
 \stackrel{(\ref{expgpbb})}{\leqslant} n \cdot \exp\left( -\frac{1}{3} \cdot \frac{(\epsilon_2)^2 \cdot (q^{IP})^2}{\mathbb{E}\big[d(s',s_i')\big]} \right) \leqslant n \cdot \exp\left( -\frac{1}{3} (\epsilon_2)^2 \cdot q^{IP} \right),
\end{equation}
where the last inequality is because of:
\[ \mathbb{E}\big[ d(s',s_i') \big] \stackrel{(\ref{exp_d_p_si})}{\leqslant} q^{LP}-d(s_{ALG}'',s_i'') \leqslant q^{IP}. \]
We want to further upperbound the probability in (\ref{pbb_error_all}). From the assumption about $\beta$ and from Lemma~\ref{size_p_star_y} we have:
\[ \frac{3R\ln(3n)}{(\epsilon_2)^2} < \beta \stackrel{\rm{Lem. \ref{size_p_star_y}}}{\leqslant} |Y| \cdot OPT \leqslant R \cdot OPT \leqslant R \cdot q^{IP}, \quad \text{equivalently}\]
\begin{equation}\label{1deltalnexp}
 \frac{1}{3} > n \cdot \exp\left( -\frac{1}{3} (\epsilon_2)^2 \cdot q^{IP} \right).
\end{equation}
So, finally we have:
\begin{equation}\label{pexistslqd}
 P\Big(\exists_{i \in \{1,2,\dots,n\}}: d(s',s_i') \geqslant q^{LP}-d(s_{ALG}'',s_i'') + \epsilon_2 \cdot q^{IP} \Big) \stackrel{(\ref{pbb_error_all}),(\ref{1deltalnexp})}{<} \frac{1}{3}.
\end{equation}
So with probability at least $\frac{2}{3}$ we obtain:
\[ \forall_{i\in\{1,2,\dots,n\}} \quad d(s' \cdot s_{ALG}'', s_i) = d(s', s_i') + d(s_{ALG}'', s_i'') \stackrel{(\ref{pexistslqd})}{<} \]
\begin{equation}\label{eps2_approx_ip}
 \stackrel{(\ref{pexistslqd})}{<} q^{LP}-d(s_{ALG}'', s_i'') +\epsilon_2\cdot q^{IP} + d(s_{ALG}'', s_i'') \leqslant (1+\epsilon_2)\cdot q^{IP}.
\end{equation}
We can also obtain a wrong number o 1's. The solution $s_{ALG}'$ for that is to take the $k'$-completion of $s'$. We will show that the additional error for such operation is not so big. Expected number of 1's in $s'$ is equal $k'$:
\[ \mathbb{E}\big[ (s')^{(1)} \big] \stackrel{\rm{def.}}{=} \mathbb{E}\left[ \sum_{j=1}^\beta s'[j] \right] \stackrel{\text{lin. of }\mathbb{E}}{=} \sum_{j=1}^\beta (s')^{LP}[j] \stackrel{\rm{def.}}{=} \big((s')^{LP}\big)^{(1)} \stackrel{(\ref{ip_no1})}{=} k'. \]
We want to know how much we lose taking the $k'$-completion. Similar as before, $(s')^{(1)} = \sum_{j=1}^\beta s'[j]$ is a sum of $\beta$ independent 0-1 variables. For $\epsilon'' \in (0,1)$ using Chernoff's bound \cite{motvani} we have:
\[ P\left( (s')^{(1)} \geqslant (1+\epsilon'')\cdot k' \right) \leqslant \exp\left( -\frac{1}{3}(\epsilon'')^2 \cdot k' \right), \]
\[ P\left( (s')^{(1)} \leqslant (1-\epsilon'')\cdot k' \right) \leqslant \exp\left( -\frac{1}{2}(\epsilon'')^2 \cdot k' \right). \]
Taking both inequalities together, $\epsilon'' = \frac{\epsilon_2}{R}$ and using assumption $k' > \frac{3R^2\ln(6)}{(\epsilon_2)^2}$ we have:
\[ P\left( \left|(s')^{(1)}-k'\right| \geqslant \epsilon''\cdot k' \right) \leqslant 2\cdot \exp\left( -\frac{1}{3}(\epsilon'')^2 \cdot k' \right) \leqslant \]
\[ \leqslant 2\cdot \exp\left( -\frac{1}{3}\frac{(\epsilon_2)^2}{R^2} \cdot k' \right) < \frac{1}{3}. \]
So with probability at least $\frac{2}{3}$ the error from taking the $k'$-completion is not greater than $ \epsilon'' \cdot k' = \frac{\epsilon_2}{R} \cdot k' \leqslant \frac{\epsilon_2}{R} \cdot \beta \stackrel{\rm{Lem. \ref{size_p_star_y}}}{\leqslant}  \frac{\epsilon_2}{R} \cdot |Y| \cdot OPT \leqslant \epsilon_2 \cdot OPT \leqslant \epsilon_2 \cdot q^{IP}$.

Combining the above with (\ref{eps2_approx_ip}) we obtain a $(1+2\epsilon_2)$-approximate solution with probability at least $\frac{1}{3}$. We may derandomize the algorithm analogously to how it was done in the PTAS for the Closest String problem  \cite{ptascs}. For more on derandomization techniques see~\cite{derandomization}. $\hfill\blacksquare$

\end{proof}

\section{Algorithm and its complexity analysis}
\label{sec:alg}
Now we are ready to combine the ideas into a single algorithm.

\begin{algorithm}
\caption{ALG(R)}
\label{alg_mav}
\begin{algorithmic}[1]
\REQUIRE $S=\{s_1,s_2,\dots,s_n\} \in (\{0,1\}^m)^n, 0 \leqslant k \leqslant m,R \in \mathbb{N}_{\geqslant 1}$
\ENSURE $s_{ALG} \in \{0,1\}^m$
\FOR{each $R$-element subset $Y=\{s_{i_1},s_{i_2},\dots,s_{i_R}\} \subseteq S$}\label{forRsubset}
  \FOR{each division $k$ into two parts $k=k'+k''$}
   \STATE{$s''_{ALG} \leftarrow k''$-completion of $(p(Y))''$ \\(if not possible, then skip this inner iteration)}
   \STATE{$s'_{ALG} \leftarrow$ approximation solution of $IP_{(\ref{ip_q})-(\ref{ip_p01})} (Y,k')$ using Lemma \ref{lem_aprox_ip}\\ (if $LP_{(\ref{ip_q})-(\ref{ip_qg0}),(\ref{lp_p01})}(Y,k')$ infeasible, then skip this inner iteration)} \label{approx_opt_problem}
   \STATE{evaluate $s'_{ALG}\cdot s''_{ALG}$ by computing $ \max_{i\in \{1,2,\dots,n\}} d(s_i,s'_{ALG}\cdot s''_{ALG})$}
  \ENDFOR
\ENDFOR\label{forRsubsetend}
\STATE{$s_{ALG} \leftarrow$ the best solution from a loop in lines \ref{forRsubset}-\ref{forRsubsetend}}
\end{algorithmic}	
\end{algorithm}
It remains to argue that for a large enough parameter $R$ the above algorithm will at some point consider a subset of votes $X$
that leads to an accurate enough approximation of the Minimax objective function of our problem.
\begin{theorem}
 $\forall_{\epsilon \in (0,1)}$ we may compute a $(1+\epsilon)$-approximate solution to Minimax Approval Voting in $O\big(Poly(n,m)\big)$ time.
\end{theorem}
\begin{proof}
 Let $\epsilon_0 = \frac{\epsilon}{3} < \frac{1}{3}$.
 
 By Lemma \ref{lem_exists_x}, there exists an $\frac{\epsilon_0 \cdot OPT}{2}$-stable set of votes $X\subseteq S$ of cardinality $|X|=R=\lceil \frac{2}{\epsilon_0} \rceil$.
 
 Consider algorithm ALG(R). In one iteration it will consider $X$ and $k',k''$ such that $(s_{OPT}')^{(1)}=k'$. Recall that $s_{ALG}''$ is the specific $k''$-completion of $\big(p(X)\big)''$. By Lemma \ref{lemma_kbiscompletion} we have: 
 \[ d(s'_{OPT} \cdot s_{ALG}'',s_i) \leqslant (1+\epsilon_0)\cdot OPT,\]
 hence $\big(s'=s_{OPT}',q=(1+\epsilon_0)\cdot OPT \big)$ is a feasible solution to $IP_{(\ref{ip_q}-\ref{ip_p01})} (X,k')$ and the optimal value of $IP_{(\ref{ip_q}-\ref{ip_p01})} (X,k')$ is at most $(1+\epsilon_0)\cdot OPT$.
 
 By Lemma \ref{lem_aprox_ip} with $\epsilon_2 = \frac{\epsilon_0}{2}$ we find a $(1+\epsilon_0)$-approximate solution $\big( s_{ALG}',q_{ALG} \big)$ to $IP_{(\ref{ip_q}-\ref{ip_p01})} (X,k')$. So we have:
 \[q_{ALG} \leqslant (1+\epsilon_0) \cdot (1+\epsilon_0) \cdot OPT \stackrel{\epsilon_0<1}{\leqslant} (1+3\epsilon_0) \cdot OPT = (1+\epsilon) \cdot OPT.\]
 
 It remains to observe, that $s_{ALG}=s_{ALG}' \cdot s_{ALG}''$ is a solution to $MAV$ of cost $q_{ALG} \leqslant (1+\epsilon) \cdot OPT$.
 
 The algorithm examined $O(n^R)=O\big(n^{\lceil \frac{6}{\epsilon} \rceil} \big) \in O(Poly(n))$ subsets $Y$, $O(m)$ choices of $k'$ and each time considered
 \[ O\left( (3n)^{\frac{108 \cdot \lceil \frac{6}{\epsilon} \rceil \cdot \ln(2) }{\epsilon^2}} + m^{\frac{108 \cdot {\lceil \frac{6}{\epsilon} \rceil}^2 \cdot \ln(6)}{\epsilon^2}} \right) \in O(Poly(n,m)) \;\, \text{cases.}\] $\hfill\blacksquare$
\end{proof}

\section{Concluding remarks}
We showed the existence of a PTAS for Minimax Approval Voting by considering all subsets of a fixed size $R$. If not the discovered supermodularity for the~inaccuracy function $ina(\cdot)$, we would simply consider all subsets of size at most $R$. Although the supermodularity was not essential for our result, it shows that larger subsets of votes are generally more stable (in the sense of definition in Lemma \ref{lem_exists_x}). It seems to suggest that an algorithm considering a smaller number of larger subsets of votes would potentially be more efficient in practice. Perhaps the most interesting open question is whether by randomly sampling a number of subsets of votes to examine, one could obtain a more practical FPRAS for the~problem.

Another interesting direction is the optimization of the Minimax objective function subject to a restriction that the voting system must be incentive compatible. According to~\cite{markakis} the best possible approximation ratio in this setting is between $2-\frac{2}{k+1}$ and $3-\frac{2}{k+1}$, and a natural challenge is to narrow this gap.

Finally, we know the complexity of the two extreme objectives, i.e., Minimax and Minisum. The latter is easily optimized by selecting the $k$ most often approved candidates. The optimization problem for intermediate objectives such as optimizing the sum of squares of the Hamming distances remains unexplored, and it would be interesting to learn which objective functions are more difficult to approximate than Minimax in the context of Approval Voting systems.

\section*{Acknowledgments}
We want to thank Katarzyna Staniewicz for many helpful proofreading comments. Also we want to thank reviewers for their valuable suggestions. Krzysztof Sornat was supported by local grant 2139/M/II/14.

\end{document}